\documentclass[12pt]{article}
\usepackage{authblk}
\usepackage{setspace} 
\usepackage[dvips]{graphicx} 
\usepackage{amsmath,amsthm,amsfonts}
\usepackage{epsfig}
\usepackage{subfig}
\RequirePackage[colorlinks,citecolor=blue,urlcolor=blue,linkcolor=blue]{hyperref}
\hypersetup{
colorlinks = true,
citecolor=blue,
urlcolor=blue,
linkcolor=blue,
pdfauthor = {Alexey Kuznetsov},
pdfpagemode = UseNone
}
    \def\qed{\hfill$\sqcap\kern-8.0pt\hbox{$\sqcup$}$\\}

    \def\re{\textnormal {Re}}
    \def\im{\textnormal {Im}}
    \def\p{{\mathbb P}}
    \def\e{{\mathbb E}}
    
    \def\r{{\mathbb R}}
    \def\c{{\mathbb C}}

    \def\d{{\textnormal d}}
    \def\i{{\textnormal i}}

\newcommand{\B}{\mathrm{B}}

	\newtheorem{theorem}{Theorem}
	
	\newtheorem{proposition}{Proposition}
	
	\newtheorem{definition}{Definition}

	\newtheorem{remark}{Remark}

\title{On ordered beta distribution and the generalized incomplete beta function}

\author{
Mayad Al-Saidi\footnote{ Dept. of Mathematics and Statistics,  York University,
4700 Keele Street, Toronto, ON, M3J 1P3, Canada. Email: mayad.alsaidi@gmail.com} ,
Alexey Kuznetsov\footnote{Dept. of Mathematics and Statistics,  York University,
4700 Keele Street, Toronto, ON, M3J 1P3, Canada.   Email: akuznets@yorku.ca}  \;  and 
Mikhail Nediak\footnote{School of Business, Queen's University, 143 Union St., Kingston, ON, K7L 3N6, Canada. 
Email: \\ mnediak@business.queensu.ca}
 }

\begin{document}

%****************************************************************************************************************
%****************************************************************************************************************
%****************************************************************************************************************

\maketitle

\begin{abstract}
Motivated by applications in Bayesian analysis we introduce a multidimensional beta distribution in an ordered simplex. We study properties of this distribution and connect them with the generalized incomplete beta function. This function is crucial in applications of multidimensional beta distribution, thus we present two efficient numerical algorithms for computing the generalized incomplete beta function, one based on Taylor series expansion and another based on Chebyshev polynomials. 
\end{abstract}

\vspace{0.25cm}

{\vskip 0.15cm}
 \noindent {\it Keywords}:  Beta distribution, beta function, incomplete beta function, ordered simplex, Chebyshev polynomials  \\
 \noindent {\it 2020 Mathematics Subject Classification }: Primary 60E05, Secondary 65D15

%****************************************************************************************************************
%****************************************************************************************************************
%****************************************************************************************************************

\section{Introduction}\label{section_intro}

Before we introduce the main obect of our study -- the ordered beta distribution -- let us consider the following motivating example. Let $S$ be a random variable with binomial distribution with parameters $N$ and $X$, where $N$ is a positive integer and $X$ is a random variable having ${\textnormal{Beta}}(\alpha,\beta)$ distribution. Then it is well-known that the posterior distribution of $X$ given $S=m$ (where $0\le m \le N$) is ${\textnormal{Beta}}(\alpha+m,\beta+N-m)$. The ordered beta distribution arises naturally in the following multidimensional generalization of the above example. Let $\{S_i\}_{1\le i \le n}$ be random variables such that $S_i$ has  binomial $(N_i,X_i)$ distribution, and  $\{X_i\}_{1\le i \le n}$ are random variables such that 
$$
0\le X_1 \le X_2 \le \dots\le X_n\le 1
$$
almost surely. What prior distribution can we assign to the vector
$$
{\mathbf X}=(X_1,X_2,\dots,X_n)
$$
such that the posterior distribution of ${\mathbf X}$ given 
$S_i=m_i$ for $i=1,2,\dots,n$ has the same form as the prior? 

To answer this question, we first define the main object of our investigation. 

\begin{definition}\label{def1}
Let $n\in {\mathbb N}$ and $a_i>0$ and $b_i>0$ for $1\le i \le n$. 
We say that a random vector ${\mathbf X}=(X_1,X_2,\dots,X_n)$ has  ordered beta distribution with parameters 
$\{a_i\}_{1\le i \le n}$ and $\{b_i \}_{1\le i \le n}$ if 
\begin{equation}\label{def_beta_distr_simplex}
\p(X_1 \in \d x_1, X_2 \in \d x_2, \dots, X_n \in \d x_n)=
C^{-1} {\mathbf 1}_{\{0\le x_1 \le x_2 \le \cdots \le x_n \le 1\}} \prod\limits_{i=1}^n x_i^{a_i-1} (1-x_i)^{b_i-1} \d x_i
\end{equation}
where $C$ is the normalization constant. 
\end{definition}

Before we proceed, we would like to explain the notation that will be used everywhere in this paper.  We denote random variables by capital letters (such as $X,Y,Z$), real numbers by lower case letters (such as $x,y,z$) and vectors by bold font (such as ${\mathbf X}$ or ${\mathbf x}$).

Definition \ref{def1} could be stated in the following  equivalent way. Take $n$ independent random variables  $Y_i   \sim {\textrm {Beta}}(a_i,b_i)$ and denote by  
$\Delta_*^n$ the ordered simplex 
\begin{equation}\label{def_simplex}
\Delta_*^n:=\{{\mathbf x}=(x_1,\dots,x_n) \in \r^n \; : \;0\le x_1 \le x_2 \le \dots \le  x_n \le 1\}. 
\end{equation}
Then the ordered beta distribution can be defined as the distribution of the vector ${\mathbf Y}=(Y_1,\dots,Y_n)$ conditioned on the event that it lies in the ordered simplex $\Delta_*^n$: 
\begin{equation}\label{def_beta_distr_simplex2}
\p({\mathbf X} \in \d {\mathbf x} )=
\p({\mathbf Y} \in \d {\mathbf x} | {\mathbf Y} \in \Delta_*^n). 
\end{equation}

% Performing a change of variables $x_i=z y_i$ we obtain an alternative integral representation for the generalized incomplete beta function:
%\begin{equation}\label{def_B_function}
%\B\Big( 
%\begin{matrix}
%a_1, \dots, a_n \\ b_1, \dots, b_n
%\end{matrix} \Big  \vert z
%\Big ) =z^{a_1+a_2+\dots+a_n} \int\limits_{\Delta_*^n} \prod\limits_{i=1}^n y_i^{a_i-1} (1-z y_i)^{b_i-1} \d y_i. 
%\end{equation}

The ordered beta distribution was introduced in \cite{LLMN_2015} 
(also subsequently used in \cite{Achtari_2018}) in the study of dynamic pricing and demand learning problems. 
Let us summarize the setup of these problems. Consider a retailer who wants to find a price level at which selling a particular product would bring the highest revenue. To solve this problem the retailer needs to know the demand for this product as a function of price. One strategy to learn the demand is to fix a set of $n$ ordered prices $p_n<p_{n-1}<\dots < p_2< p_1$ and then to estimate the proportion of population that would buy the product at price $p_j$ (for each $j$). Let $X_j\in (0,1)$ be the probability that a randomly sampled customer would buy the product at price $p_j$. The demand is a decreasing function of price, which implies that the probabilities $X_j$ must satisfy  
$0\le X_1 \le X_2 \le \dots \le X_n \le 1$. We observe how many customers made the decision to purchase or not to purchase the product priced at $p_j$, and we denote by $m_j$ the number of successful sales at price $p_j$ while $k_j$ is the number of times the product was offered at price $p_j$ but the customer decided not to purchase it. We assume that the customers make decisions independently of each other, thus the number of customers who decided to purchase the product at price level $p_j$ had binomial distribution with parameters $m_j+k_j$ and $X_j$. The question is how the retailer can use the information contained in vectors 
${\mathbf m}=(m_1,\dots,m_n)$ and ${\mathbf k}=(k_1,\dots,k_n)$
to learn more about the probabilities $X_j$ and to increase the revenue? The answer to this question can be found in \cite{LLMN_2015}, and it is based on the following easily-verified property:   if the prior distribution of purchase probabilities ${\mathbf X}=(X_1,\dots,X_n)$ is ordered beta  with parameters ${\mathbf a}=(a_1,\dots,a_n)$ and ${\mathbf b}=(b_1,\dots,b_n)$, then the posterior distribution of ${\mathbf X}$ given ${\mathbf m}$ and ${\mathbf k}$ is also ordered beta with parameters ${\mathbf a}+{\mathbf m}$ and ${\mathbf b}+{\mathbf k}$. This result can also be found in \cite{Achtari_2018} [Proposition 1, page 13].

The aim of this paper is to study the ordered beta distribution and to state its properties and also to present several numerical algorithms which will make it easier to use this distribution in applications. The paper is organized as follows. In Section \ref{section2} we present various probabilistic properties of ordered beta distribution (such as marginal distributions). To state these properties, we need to introduce a new special function, which we call {\it the generalized incomplete beta function}. The generalized incomplete beta function is indispensable for applications of the ordered beta distribution, in particular it is important to be able to compute this function numerically. In Section \ref{section3} we study analytic properties of the generalized incomplete beta function and in Section \ref{section4} we provide two efficient numerical algorithms for its computation. We demonstrate the efficiency of these algorithms by presenting the results of several numerical experiments in Section \ref{section5}.

\section{Properties of the ordered beta distribution}\label{section2}

The following function will be needed for describing the properties of the ordered beta distribution: 
\begin{definition}\label{def1}
{
Let $n\in {\mathbb N}$ and $a_i>0$, $b_i>0$ for $1\le i \le n$. 
 {\rm The generalized incomplete beta function} (with parameters 
$\{a_i\}_{1\le i \le n}$ and $\{b_i \}_{1\le i \le n}$) is defined for $z\in [0,1]$  as follows:
\begin{equation}\label{def_B_function}
\B\Big( 
\begin{matrix}
a_1, \dots, a_n \\ b_1, \dots, b_n
\end{matrix} \Big  \vert z
\Big ) =\int\limits {\mathbf 1}_{\{0 \le x_1 \le x_2 \le \dots \le x_n \le z\}}  \Big[ \prod\limits_{i=1}^n x_i^{a_i-1} (1-x_i)^{b_i-1} \Big] \d x_1 \dots \d x_n. 
\end{equation}
We will call 
$\B\Big( 
\begin{matrix}
a_1, \dots, a_n \\ b_1, \dots, b_n
\end{matrix} \Big ):=\B\Big( 
\begin{matrix}
a_1, \dots, a_n \\ b_1, \dots, b_n
\end{matrix} \Big  \vert 1
\Big )
$
{\rm the generalized beta function}.} 
\end{definition}

It is clear that the normalization constant in \eqref{def_beta_distr_simplex} must be given by 
$
C=\B\Big( 
\begin{matrix}
a_1, \dots, a_n \\ b_1, \dots, b_n
\end{matrix} \Big ).
$
When $n=1$ we recover the classical incomplete beta function
$$
\B\Big( 
\begin{matrix}
a_1 \\ b_1
\end{matrix} \Big  \vert z
\Big )=\int_0^z x^{a_1-1} (1-x)^{b_1-1} \d x, 
$$
and the classical beta function
$$
\B\Big( 
\begin{matrix}
a_1 \\ b_1
\end{matrix} \Big)=\int_0^1 x^{a_1-1} (1-x)^{b_1-1} \d x=\frac{\Gamma(a_1)\Gamma(b_1)}{\Gamma(a_1+b_1)}. 
$$
The above two functions are usually denoted as 
$\B_z(a,b)$ and $\B(a,b)$, see  \cite{Jeffrey2007}[Section 8.38].

Everywhere in this section we assume that $n\ge 2$ and that 
${\mathbf X}$ has ordered beta distribution with parameters 
$\{a_i\}_{1\le i \le n}$ and $\{b_i \}_{1\le i \le n}$. 
To simplify the presentation of results, in the rest of the paper we follow the convention that 
$\B\Big( 
\begin{matrix}
- \\ -
\end{matrix} \Big | z \Big )=1$ for all $z$. 

In the next result we collect several properties of the ordered beta distriubution. 

\begin{theorem}\label{Thm1}
Let ${\mathbf X}=(X_1,X_2,\dots,X_n)$ have ordered beta distribution with parameters $\{a_i\}_{1\le i \le n}$ and $\{b_i \}_{1\le i \le n}$ and let 
$C=\B\Big( 
\begin{matrix}
a_1, \dots, a_n \\ b_1, \dots, b_n
\end{matrix} \Big )$. The following statements are true: 
\begin{itemize}
\item[(i)] The random vector
 $\hat {\mathbf X}:=(1-X_n,1-X_{n-1},\dots,1-X_1)$ has 
ordered beta distribution with parameters 
$(b_n,b_{n-1},\dots,b_1)$ and $(a_n,a_{n-1},\dots,a_1)$. 
\item[(ii)]
For $1\le k \le n-1$ and $z\in (0,1)$
\begin{equation}\label{Thm1_eqn1}
\p(X_k\le z<X_{k+1})= C^{-1} \B\Big( 
\begin{matrix}
a_1, \dots, a_k \\ b_1, \dots, b_k
\end{matrix} \Big  \vert z
\Big )
\B\Big( 
\begin{matrix}
b_{n}, b_{n-1}, \dots, b_{k+1} \\ a_n, a_{n-1}, \dots, a_{k+1}
\end{matrix} \Big  \vert 1-z
\Big ).
\end{equation}
\item[(iii)] For $1\le k \le n$ and $x\in (0,1)$
\begin{equation}\label{Thm1_eqn2}
\p(X_k \in \d x)=C^{-1} x^{a_k-1} (1-x)^{b_k-1} 
\B\Big( 
\begin{matrix}
a_1, \dots, a_{k-1} \\ b_1, \dots, b_{k-1}
\end{matrix} \Big  \vert x
\Big )
\B\Big( 
\begin{matrix}
b_{n}, b_{n-1}, \dots, b_{k+1} \\ a_n, a_{n-1}, \dots, a_{k+1}
\end{matrix} \Big  \vert 1-x
\Big ) \d x.
\end{equation}
\item[(iv)] For $1\le k \le n$ and $z\in [0,1]$
\begin{align}\label{Thm1_eqn3}
\p(X_k \le z)=
 \sum\limits_{j=k}^{n} C^{-1} \B\Big( 
\begin{matrix}
a_1, \dots, a_j \\ b_1, \dots, b_j
\end{matrix} \Big  \vert z
\Big )
\B\Big( 
\begin{matrix}
b_{n}, b_{n-1}, \dots, b_{j+1} \\ a_n, a_{n-1}, \dots, a_{j+1}
\end{matrix} \Big  \vert 1-z
\Big ),
\end{align}
and 
\begin{align}\label{Thm1_eqn4}
\p(z < X_k)&=
 \sum\limits_{j=0}^{k-1} C^{-1} \B\Big( 
\begin{matrix}
a_1, \dots, a_j \\ b_1, \dots, b_j
\end{matrix} \Big  \vert z
\Big )
\B\Big( 
\begin{matrix}
b_{n}, b_{n-1}, \dots, b_{j+1} \\ a_n, a_{n-1}, \dots, a_{j+1}
\end{matrix} \Big  \vert 1-z
\Big ).
\end{align}
\item[(v)]
If $\alpha>-a_k$ and $\beta>-b_k$ then 
\begin{equation}\label{Thm1_eqn5}
\e[X_k^{\alpha} (1-X_k)^{\beta}]= C^{-1} \B\Big( 
\begin{matrix}
a_1, \dots,a_k+\alpha,\dots a_n \\ b_1, \dots,b_k+\beta,\dots, b_n
\end{matrix} 
\Big ),
\end{equation}
and, more generally, if $\alpha_i>-a_i$ and $\beta_i>-b_i$ for $i=1,2,\dots,n$ then
\begin{equation}\label{Thm1_eqn6}
\e\Big[\prod\limits_{i=1}^n X_i^{\alpha_i} (1-X_k)^{\beta_i}]= C^{-1} \B\Big( 
\begin{matrix}
a_1+\alpha_1, \dots, a_n+\alpha_n \\ b_1+\beta_1,\dots, b_n+\beta_n
\end{matrix} 
\Big ).
\end{equation}

\end{itemize}

\end{theorem}
\begin{proof}
The result in item (i) follows directly from Definition 
\ref{def1}. To prove the result in item (ii) we first note the following identity: for $z \in (0,1)$ 
\begin{equation}\label{1_minus_z_symmetry}
\int_{[0,1]^n} {\mathbf 1}_{\{z\le x_1 \le x_2 \le \cdots \le x_n \le 1\}} \prod\limits_{i=1}^n x_i^{a_i-1} (1-x_i)^{b_i-1} \d x_i
=\B\Big( 
\begin{matrix}
b_n, b_{n-1} \dots, b_1 \\ a_{n}, a_{n-1}, \dots, a_1
\end{matrix} \Big  \vert 1-z
\Big ).
\end{equation}
Formula \eqref{1_minus_z_symmetry} can be derived by applying a change of variables $x_i \mapsto 1-y_i$ and using \eqref{def_B_function}. 
Now we compute with the help of 
\eqref{def_B_function}
and \eqref{1_minus_z_symmetry} 
\begin{align*}
\p(X_k \le z<X_{k+1})&= 
\int_{[0,1]^n} C^{-1} {\mathbf 1}_{\{0\le x_1 \le x_2 \le \cdots \le x_n \le 1\}} 
{\mathbf 1}_{\{x_k<z<x_{k+1}\}} \prod\limits_{i=1}^n x_i^{a_i-1} (1-x_i)^{b_i-1} \d x_i\\
&=
C^{-1}\int_{[0,1]^k} {\mathbf 1}_{\{0\le x_1 \le x_2 \le \cdots \le x_k < z\}} 
\prod\limits_{i=1}^k x_i^{a_i-1} (1-x_i)^{b_i-1} \d x_i\\
& \;\;\;\;\;\; \times 
\int_{[0,1]^{n-k}} {\mathbf 1}_{\{z < x_{k+1} \le \cdots \le x_n \le 1\}} 
\prod\limits_{i=k+1}^n x_i^{a_i-1} (1-x_i)^{b_i-1} \d x_i \\
&=
C^{-1} \B\Big( 
\begin{matrix}
a_1, \dots, a_k \\ b_1, \dots, b_k
\end{matrix} \Big  \vert z
\Big )
\B\Big( 
\begin{matrix}
b_{n}, b_{n-1}, \dots, b_{k+1} \\ a_n, a_{n-1}, \dots, a_{k+1}
\end{matrix} \Big  \vert 1-z
\Big ). 
\end{align*}
The derivation of the marginal distribution of $X_k$ given in \eqref{Thm1_eqn2} follows in the same way from  \eqref{def_beta_distr_simplex} by integrating over variables all $x_i$ with $i\neq k$ and using \eqref{def_B_function}
and \eqref{1_minus_z_symmetry} . 

Formula \eqref{Thm1_eqn3} follows from \eqref{Thm1_eqn1} by noting that for an ordered beta distributed vector ${\mathbf X}$ the event $\{X_k \le z\}$ is the disjoint union of events 
$\{X_j \le z<X_{j+1}\}$, $k\le j \le n$. An identical argument is used to prove \eqref{Thm1_eqn4}. 

Formulas \eqref{Thm1_eqn5} and \eqref{Thm1_eqn6} follow directly from \eqref{def_beta_distr_simplex} and \eqref{def_B_function}.
\end{proof}

\section{Properties of the generalized incomplete Beta function}\label{section3}

In the next proposition we collect several properties of the generalized incomplete beta function, which will be useful later. 

\begin{proposition}\label{Proposition1}
Assume that $z\in (0,1)$ and $a_i>0$, $b_i>0$ for $1\le i \le n$. Then the following statements are true:
\begin{itemize}
\item[(i)] 
\begin{equation}\label{B_symmetry}
\B\Big( 
\begin{matrix}
a_1, a_2, \dots, a_n \\ b_1, b_2, \dots, b_n
\end{matrix} \Big )=\B\Big( 
\begin{matrix}
 b_n, b_{n-1}, \dots, b_1 \\
 a_n, a_{n-1}, \dots, a_1 
\end{matrix} \Big ), 
\end{equation}
\item[(ii)]
\begin{align}\label{B_sum_identity}
\sum\limits_{k=0}^{n} 
 \B\Big( 
\begin{matrix}
a_1, \dots, a_k \\ b_1, \dots, b_k
\end{matrix} \Big  \vert z
\Big )
\B\Big( 
\begin{matrix}
b_{n}, b_{n-1}, \dots, b_{k+1} \\ a_n, a_{n-1}, \dots, a_{k+1}
\end{matrix} \Big  \vert 1-z
\Big )=\B\Big( 
\begin{matrix}
a_1, \dots, a_n \\ b_1, \dots, b_n
\end{matrix} \Big  ),
\end{align}
\item[(iii)]
\begin{align}\label{B_sum_identity2}
\sum\limits_{k=0}^{n} (-1)^k
 \B\Big( 
\begin{matrix}
a_1, \dots, a_k \\ b_1, \dots, b_k
\end{matrix} \Big  \vert z
\Big )
\B\Big( 
\begin{matrix}
a_{n}, a_{n-1}, \dots, a_{k+1} \\ b_n, b_{n-1}, \dots, b_{k+1}
\end{matrix} \Big  \vert z
\Big )=0,
\end{align}
\item[(iv)]
\begin{equation}\label{B_integral_identity}
\B\Big( 
\begin{matrix}
a_1, a_2, \dots, a_n \\ b_1, b_2, \dots, b_n
\end{matrix} \Big  \vert z
\Big )=\int_0^z x^{a_n-1} (1-x)^{b_n-1} 
\B\Big( 
\begin{matrix}
a_1, a_2, \dots, a_{n-1} \\ b_1, b_2, \dots, b_{n-1}
\end{matrix} \Big  \vert x
\Big ) \d x,
\end{equation}

\item[(v)] 
for $1\le k \le n$
\begin{equation}\label{B_integral_identity2}
\B\Big( 
\begin{matrix}
a_1, a_2, \dots, a_n \\ b_1, b_2, \dots, b_n
\end{matrix} \Big )=\int_0^1 x^{a_k-1} (1-x)^{b_k-1} 
\B\Big( 
\begin{matrix}
a_1, a_2, \dots, a_{k-1} \\ b_1, b_2, \dots, b_{k-1}
\end{matrix} \Big | x \Big )
\B\Big( 
\begin{matrix}
b_n, b_{n-1}, \dots, b_{k+1} \\ a_n, a_{n-1}, \dots, a_{k+1}
\end{matrix} \Big | 1-x \Big ) \d x.
\end{equation}

\end{itemize}
\end{proposition}
\begin{proof}
Formula \eqref{B_symmetry} follows from \eqref{def_B_function} by setting $z=1$ and changing the variables of integration $x_i \mapsto 1-y_i$. The identity \eqref{B_sum_identity} is obtained from 
\eqref{Thm1_eqn3} and \eqref{Thm1_eqn4} by noting that $\p(X_k \le z) + \p(z<X_k)=1$. 

To prove formula \eqref{B_sum_identity2} one would need to multiply both sides of the following identity
\begin{align}\label{indicator_identity}
&{\mathbf 1}_{\{  x_1 \le x_2 \le  \cdots \le x_n \}}
- {\mathbf 1}_{\{ x_1 \le x_2 \le  \cdots \le x_{n-1}   \}}
+{\mathbf 1}_{\{ x_1 \le x_2 \le  \cdots \le x_{n-2}  \}}{\mathbf 1}_{\{ x_n < x_{n-1} \}} \\
\nonumber
&-{\mathbf 1}_{\{ x_1 \le x_2 \le  \cdots \le x_{n-3}  \}}{\mathbf 1}_{\{ x_n < x_{n-1} < x_{n-2} \}}
+\dots+(-1)^n {\mathbf 1}_{\{x_n < x_{n-1}  < \cdots < x_{1} \}}=0, 
\end{align}
by 
$$
\prod\limits_{i=1}^n x_i^{a_i-1} (1-x_i)^{b_i-1} \d x_i
$$
and integrate over $[0,z]^n$. The identity \eqref{indicator_identity} can be established easily by induction; the case $n=2$ is equivalent to an obvious statement
$$
{\mathbf 1}_{\{x_1 \le x_2\}}+{\mathbf 1}_{\{x_2<x_1\}}=1. 
$$

Formula \eqref{B_integral_identity} follows from the definition \eqref{def_B_function} and formula 
\eqref{B_integral_identity2} is a corollary of \eqref{Thm1_eqn2}. 
\end{proof}

For $a_i>0$, $b_i>0$ and $z\in (0,1]$ we define
\begin{equation}\label{def_ beta}
\beta\Big( 
\begin{matrix}
a_1, \dots, a_n \\ b_1, \dots, b_n
\end{matrix} \Big  \vert z
\Big ):=
z^{-a_1-a_2-\dots-a_n} \B\Big( 
\begin{matrix}
a_1, \dots, a_n \\ b_1, \dots, b_n
\end{matrix} \Big  \vert z
\Big ), 
\end{equation}
This function has better analytical properties, compared with the generalized incomplete Beta function.

\begin{proposition}\label{prop_beta_analytic}
The function 
$z\mapsto \beta\Big( 
\begin{matrix}
a_1, \dots, a_n \\ b_1, \dots, b_n
\end{matrix} \Big  \vert z
\Big )$  is analytic in ${\mathbb C} \setminus [1,\infty)$ and satisfies 
\begin{equation}\label{beta_integral_identity}
\beta\Big( 
\begin{matrix}
a_1, a_2, \dots, a_n \\ b_1, b_2, \dots, b_n
\end{matrix} \Big  \vert z
\Big )=\int_0^1 y^{a_1+a_2+\dots+a_n-1} (1-yz)^{b_n-1} 
\beta\Big( 
\begin{matrix}
a_1, a_2, \dots, a_{n-1} \\ b_1, b_2, \dots, b_{n-1}
\end{matrix} \Big  \vert yz
\Big ) \d y.
\end{equation}
If $b_i \in {\mathbb N}$ for all $1\le i \le n$, then   
$
\beta\Big( 
\begin{matrix}
a_1, \dots, a_n \\ b_1, \dots, b_n
\end{matrix} \Big  \vert z
\Big )$ is a polynomial in $z$. 
%and for $0\le z\le 1/2$ we have 
%$$
%\beta\Big( 
%\begin{matrix}
%a_1, a_2, \dots, a_n \\ b_1, b_2, \dots, b_n
%\end{matrix} \Big  \vert z
%\Big )< \prod\limits_{j=1}^n \max(1,2^{1-b_j})/A_j
%$$
%where $A_j=\sum\limits_{l=1}^j a_l$. 
\end{proposition}
\begin{proof}
For $z\in (0,1)$ formula \eqref{beta_integral_identity} follows from \eqref{B_integral_identity} by a change of variables 
$x=zy$. The fact that $\beta\Big( 
\begin{matrix}
a_1, \dots, a_n \\ b_1, \dots, b_n
\end{matrix} \Big  \vert z
\Big )$ is analytic in ${\mathbb C} \setminus [1,\infty)$ follows from \eqref{beta_integral_identity} by induction on $n$ (note that for every $y\in (0,1)$ the function $z\mapsto (1-yz)^{b_n-1}$ is analytic in ${\mathbb C} \setminus [1,\infty)$). 
The fact that $\beta$-function is a polynomial in $z$ when the coefficients $b_i$ are integers follows from 
\eqref{beta_integral_identity} by induction in $n$. 
\end{proof}

It is clear from Proposition \ref{prop_beta_analytic} that the generalized incomplete beta function $\B\Big( 
\begin{matrix}
a_1, \dots, a_n \\ b_1, \dots, b_n
\end{matrix} \Big  \vert z
\Big )$ is a polynomial in $z$ if  $a_1+a_2+\dots+a_n$ and all coefficients $b_i$ are positive integers and that in the general case it is an analytic function in the complex plane with two cuts along $(-\infty,0]$ and $[1,\infty)$.

\section{Computing the generalized incomplete Beta function}\label{section4}

Now we turn our attention to the question of computing the generalized incomplete beta function for arbitrary $z\in (0,1]$. 
Our first observation is that it is enough to be able to compute this function for $z\in (0,1/2]$. To show this, we rewrite the identity \eqref{B_sum_identity} in the form  
\begin{align}\label{B_sum_identity_new}
\B\Big( 
\begin{matrix}
a_1, \dots, a_n \\ b_1, \dots, b_n
\end{matrix} \Big  )&=
\B\Big( 
\begin{matrix}
a_1, \dots, a_n \\ b_1, \dots, b_n
\end{matrix} \Big  \vert z
\Big )+
\B\Big( 
\begin{matrix}
b_{n}, b_{n-1}, \dots, b_{1} \\ a_n, a_{n-1}, \dots, a_{1}
\end{matrix} \Big  \vert 1-z
\Big )\\ \nonumber
&+
\sum\limits_{k=1}^{n-1} 
 \B\Big( 
\begin{matrix}
a_1, \dots, a_k \\ b_1, \dots, b_k
\end{matrix} \Big  \vert z
\Big )
\B\Big( 
\begin{matrix}
b_{n}, b_{n-1}, \dots, b_{k+1} \\ a_n, a_{n-1}, \dots, a_{k+1}
\end{matrix} \Big  \vert 1-z
\Big ). 
\end{align} 
We see that we can compute the value of 
$\B\Big( 
\begin{matrix}
a_1, \dots, a_n \\ b_1, \dots, b_n
\end{matrix} \Big  )$
by setting $z=1/2$ in \eqref{B_sum_identity_new} and when $z\in (1/2,1)$  we can compute 
$\B\Big( 
\begin{matrix}
a_1, \dots, a_n \\ b_1, \dots, b_n
\end{matrix} \Big  \vert z
\Big )$ by induction, since all other terms in 
\eqref{B_sum_identity_new} either depend on $1-z \in (0,1/2)$ or have fewer parameters. Thus, one can use identity \eqref{B_sum_identity_new} coupled with induction on $n$ to compute the value of any generalized incomplete beta function (with arbitrary number of parameters) for any value of $z\in [0,1]$. 

Instead of computing the  generalized incomplete beta function
$B\Big( 
\begin{matrix}
a_1, \dots, a_n \\ b_1, \dots, b_n
\end{matrix} \Big  \vert z
\Big )$, it is more convenient to solve an equivalent problem of computing
$\beta\Big( 
\begin{matrix}
a_1, \dots, a_n \\ b_1, \dots, b_n
\end{matrix} \Big  \vert z
\Big )$, as this function is analytic in a wider domain. Next we present two algorithms for computing this function:  the first algorithm is based on Taylor series expansion and the second algorithm is based on Chebyshev series.

\subsection{The Taylor expansion method}\label{subsection_computing_Taylor}

The computation proceeds by iteration of 
\eqref{B_integral_identity}. 
First of all, we note that when $n=1$ we have a series representation 
\begin{equation}\label{incomplete_beta_Taylor}
\beta\Big( 
\begin{matrix}
a_1 \\ b_1
\end{matrix} \Big  \vert z
\Big )=\int_0^1 y^{a_1-1} (1-yz)^{b_1-1} \d y= \sum\limits_{k\ge 0}\frac{(1-b_1)_k}{(a_1+k) k!} z^{k}, \;\;\; |z|<1, 
\end{equation}
which is easily obtained by integrating the binomial series expansion of the term $(1-yz)^{b_1}$. Here $(a)_k=a(a+1)\dots(a+k-1)$ denotes the Pochhammer symbol. 
Computing the coefficients of Taylor series of $\beta\Big( 
\begin{matrix}
a_1, \dots, a_n \\ b_1, \dots, b_n
\end{matrix} \Big  \vert z
\Big )$ when $n\ge 2$ can be done iteratively with the help of the following result. 

\begin{proposition}
For $n\ge 1$ and $k\ge 0$ denote by  $c_k^{(n)}$ the coefficients in the Taylor  series expansions 
\begin{equation}\label{beta_Taylor_series}
\beta\Big( 
\begin{matrix}
a_1, \dots, a_n \\ b_1, \dots, b_n
\end{matrix} \Big  \vert z
\Big )=
 \sum\limits_{k\ge 0}  c^{(n)}_k z^{k}, \;\;\; |z|<1.
\end{equation}
Then for $n\ge 2$ and $k\ge 0$
\begin{equation}\label{c_k_n_recursive_formula}
c^{(n)}_k=\frac{1}{k+a_1+\dots+a_n} \sum\limits_{l=0}^k c^{(n-1)}_{k-l} \times \frac{(1-b_n)_l}{l!} . 
\end{equation}
\end{proposition}
\begin{proof}
Using the Binomial series and \eqref{B_integral_identity} we obtain for $|z|<1$
\begin{align*}
\beta\Big( 
\begin{matrix}
a_1, \dots, a_n \\ b_1, \dots, b_n
\end{matrix} \Big  \vert z
\Big )&=\int_0^1 y^{a_1+\dots+a_n-1} (1-yz)^{b_n-1} \beta\Big( 
\begin{matrix}
a_1, \dots, a_{n-1} \\ b_1, \dots, b_{n-1}
\end{matrix} \Big  \vert yz
\Big ) \d y\\ \nonumber
& = \int_0^1 y^{a_1+\dots+a_n-1} \Big[ \sum\limits_{l\ge 0} \frac{(1-b_n)_l}{l!} (yz)^l \Big] 
\times \Big[  \sum\limits_{k\ge 0} c_k^{(n-1)} (yz)^{k} \Big]\d y\\
& = \int_0^1   y^{a_1+\dots+a_{n}-1} 
\sum\limits_{k\ge 0}  \Big[ \sum\limits_{l=0}^k c^{(n-1)}_{k-l} \times \frac{(1-b_n)_l}{l!} \Big] (yz)^{k} \d y\\
&=\sum\limits_{k\ge 0} \frac{z^k}{k+a_1+\dots+a_n} \Big[\sum\limits_{l=0}^k c^{(n-1)}_{k-l} \times \frac{(1-b_n)_l}{l!}\Big].
\end{align*}
\end{proof}

Now we present the algorithm for computing 
$\beta\Big( 
\begin{matrix}
a_1, \dots, a_n \\ b_1, \dots, b_n
\end{matrix} \Big  \vert z
\Big )$
for $z\in (0,1/2]$.  

\vspace{0.25cm} 
\noindent 
{\bf Step 1:} Fix a large integer $N$ and compute  $\{c_k^{(1)}\}_{0\le k \le N}$ via \eqref{incomplete_beta_Taylor}.

\vspace{0.25cm} 
\noindent 
{\bf Step 2:} For $m=2,3,\dots,n$, once we have the values of $\{c_k^{(m-1)}\}_{0\le k \le N}$ we compute the values of $\{c_k^{(m)}\}_{0\le k \le N}$ via 
\eqref{c_k_n_recursive_formula}.  This computation involves a convolution of two sequences, thus each step can be made more efficient with the help of the Fast Fourier Transform. 

\vspace{0.25cm} 
\noindent 
{\bf Step 3:}
Compute the value of 
$\beta\Big( 
\begin{matrix}
a_1, \dots, a_n \\ b_1, \dots, b_n
\end{matrix} \Big  \vert z
\Big )$
by truncating the series in \eqref{incomplete_beta_Taylor} after $N$ terms. 

\vspace{0.25cm} 

Several remarks are in order. First of all, we need to discuss the error due to truncation of series in \eqref{incomplete_beta_Taylor}.  
Since the function 
$\beta\Big( 
\begin{matrix}
a_1, \dots, a_n \\ b_1, \dots, b_n
\end{matrix} \Big  \vert z
\Big )$ is analytic in ${\mathbb C}\setminus [1,\infty)$ and is continuous as $z\to 1$, for every $n$ we have $c_k^{(n)} \to 0$ as $k \to +\infty$. This fact can be seen by indentifying these coefficients with the coefficients of Fourier series of a continuous function 
$$
t\mapsto \beta\Big( 
\begin{matrix}
a_1, \dots, a_n \\ b_1, \dots, b_n
\end{matrix} \Big  \vert e^{2\pi \i t}
\Big ).$$
 Thus after truncating the series in \eqref{incomplete_beta_Taylor} after $N$ terms we will have an error of $O(|z|^{N})$, and since $z\in (0,1/2]$ this error is less than $C2^{-N}$, for some constant $C$, which may depend on $a_i$ and $b_i$. Second, let us discuss the computational complexity of this algorithm. Using the Fast Fourier Transform, each Step 2 can be computed in $O(N \ln(N))$ arithmetic operations. 
Thus, the computational complexity of evaluating $\beta\Big( 
\begin{matrix}
a_1, \dots, a_n \\ b_1, \dots, b_n
\end{matrix} \Big  \vert z
\Big )$
is $O(n N \ln(N))$ arithmetic operations. 

The previous two comments show that the algorithm based on Taylor series is fast-convergent and very efficient. However,  there is one potential problem with this algorithm. When one or more of parameters $a_i$ (or $b_i$) is large, the coefficients 
$(1-a_i)_l$ (or $(1-b_i)_l$) that will appear in \eqref{c_k_n_recursive_formula} will also become very large and will have alternating sign, causing lot of cancellation in the sum in \eqref{c_k_n_recursive_formula} and resulting in loss of precision. Therefore, when some of the parameters $a_i$, $b_i$ are large one has to be mindful of this potential loss of precision. We resolved this problem by using a multi-precision arithmetic when computing the coefficients $c_k^{(n)}$ via \eqref{c_k_n_recursive_formula}.

\subsection{The Chebyshev expansion method}\label{subsection_computing_Chebyshev}

Let $\{\xi^{(n)}_k\}_{0\le k \le N}$ be the coefficients in the Chebyshev expansion 
\begin{equation}\label{expansion_beta_n}
\beta\Big( 
\begin{matrix}
a_1, \dots, a_{n} \\ b_1, \dots, b_{n}
\end{matrix} \Big  \vert z
\Big )=\frac{1}{2}\xi^{(n)}_0 + \sum\limits_{k\ge 1} \xi^{(n)}_k T_k(4z-1).
\end{equation}
Note that we use scaled Chebyshev polynomials $T_k(4z-1)$, since we are interested in computing the $\beta$-function only for $z\in (0,1/2]$ and the function $4z-1$ maps the interval $(0,1/2]$ onto $(-1,1)$. Below we present an algorithm that computes the coefficients $\{\xi^{(m)}_k\}_{0\le k \le N}$ using the previously computed values of 
$\{\xi^{(m-1)}_k\}_{0\le k \le N}$.

%Assume we know the coefficients $\{d_n\}_{0\le n \le N}$ in the Chebyshev expansion of 
%$$
%f(z)=d_0/2+\sum\limits_{n\ge 1} d_n T_n(4z-1). 
%$$
%
%The problem is to find efficiently and accurately the Chebyshev coefficients of the function 
%\begin{equation}\label{eqnFz}
%F(z):=z^{-a}\int_0^z x^{a-1}(1-x)^{b-1} f(x) dx, \;\;\; 0< z \le 1/2.
%\end{equation}
%More precisely, we need to approximate the numbers $\{g_n\}_{0\le n \le N}$ appearing in the Chebyshev expansion
%\begin{equation}\label{Fz_expansion}
%F(z) = g_0/2+\sum\limits_{n\ge 1} g_n T_n(4z-1).
%\end{equation}

\vspace{0.25cm} 
\noindent 
{\bf Initial step:} Choose a large integer $N$ and set $m=1$. Let $\hat\xi^{(0)}_0=2$ and $\hat\xi^{(0)}_k=0$ for $1\le k\le N$.  

\vspace{0.25cm} 
\noindent 
{\bf Step 1:}
For $j=0,1,\dots,N$ compute (using the Fast Fourier Transform)
\begin{equation}\label{eqn_vj}
v_j=\frac{1}{2}\hat\xi^{(m-1)}_0 + \sum\limits_{k=1}^N \hat\xi^{(m-1)}_k \cos\Big(\frac{\pi k (j+1/2)}{N+1}\Big).
\end{equation}

\vspace{0.25cm} 
\noindent 
{\bf Step 2:}
For $k=0,1,\dots,N$ compute (using the Fast Fourier Transform)
\begin{equation}\label{eqn_etak}
\eta_k=\frac{2}{N+1}\sum\limits_{j=0}^N (1-(1+x_j)/4)^{b_m-1} v_j \cos\Big(\frac{\pi k (j+1/2)}{N+1}\Big),
\end{equation}
where 
$$
x_j:=\cos\Big(\frac{\pi (j+1/2)}{N+1}\Big). 
$$

\vspace{0.25cm} 
\noindent 
{\bf Step 3:}
We set $\mu_{N}=\mu_{N+1}=0$ and $A_m:=a_1+a_2+\dots+a_m$ and compute recursively
\begin{equation}\label{c_n-1_recusion}
\mu_{k-1}= \frac{1}{1+A_m/k}\Big(8\eta_k-2 \mu_k-\mu_{k+1} (1-A_m/k) \Big), \;\;\; k=N, N-1, \dots, 1,
\end{equation}

\vspace{0.25cm} 
\noindent 
{\bf Step 4:} We compute 
\begin{equation}\label{formula_g0}
\hat \xi^{(m)}_{0}=\frac{1}{A_m}\Big(\eta_0-\frac{1}{4}(\mu_0+\mu_1)\Big)
\end{equation}
and 
\begin{equation}\label{formula_g_n}
\hat \xi^{(m)}_{k}=\frac{\mu_{k-1}-\mu_{k+1}}{8k}, \;\;\; k=1,2,\dots,N.
\end{equation}

\vspace{0.25cm}

After running the above algorithm for $m=1,2,\dots,n$, we obtain the coefficients $\{\hat \xi^{(n)}_k\}_{0\le k \le N}$ (which also depend on $N$) that serve as approximations to the desired coefficients $\xi^{(n)}_k$. It is clear that the computational complexity of this algorithm is also $O(n N \ln(N))$ -- the same as for Taylor expansion method. The next theorem gives us a bound for the error of the approximation.

\begin{theorem}\label{thm_Chebyshev_bound}
There exists a constant $C$ (depending on $a_i$ and $b_i$) such that for all $N\ge 1$
\begin{equation}\label{eqn_xi_n_k_bound}
\max\limits_{0\le k \le N} |\hat\xi^{(n)}_{k}-\xi^{(n)}_{k}|<C 3^{-8N/5}.
\end{equation}
\end{theorem}

\begin{remark}
The constant in \eqref{eqn_xi_n_k_bound} can be slightly improved: as will be clear from the proof of Theorem \ref{thm_Chebyshev_bound}, it is true that  
$$
\max\limits_{0\le k \le N}|\hat\xi^{(n)}_{k}-\xi^{(n)}_{k}|=O(\rho^{-N})
$$
for every $\rho \in (0,3+8\sqrt{8})$, where the constant in the big-O notation depends on $a_i$, $b_i$ and $\rho$. Note that $3^{8/5}=5.7995...$ and $3+\sqrt{8}=5.8284...$, thus the bound given in \eqref{eqn_xi_n_k_bound}, while suboptimal, should be good enough for all practical purposes. 
\end{remark}

The proof of Theorem \ref{thm_Chebyshev_bound} relies on the following 

\vspace{0.25cm}
{\bf Fact:}
{\it Assume the function $g(z)$ is analytic in $\c \setminus [1,\infty)$. Then the coefficients in the Chebyshev expansion 
$$
g(z)=w_0/2+\sum\limits_{k\ge 1} w_k T_k(4z-1), \;\;\; 0\le z \le 1/2,
$$
satisfy $w_k=O(\rho^{-k})$ for any $\rho \in (0,3+\sqrt{8})$. } 

\vspace{0.25cm}
The above result follows easily from the classical fact that the Chebyshev series expansion of a function $f$ converges in the largest Bernstein ellipse in which the function $f$ is analytic (see \cite{Mason_2022}[Theorem 5.16]). More precisely, the function $f((1+z)/4)$ is analytic in $\c \setminus [3,\infty)$, thus it is analytic inside the  Bernstein ellipse 
$$
E_{\rho}:=\{ z \in \c : \frac{\re(z)^2}{a_{\rho}^2}+\frac{\im(z)^2}{b_{\rho}^2}=1 \}
$$
with $a_{\rho}=(\rho+\rho^{-1})/2$ and $b_{\rho}=(\rho-\rho^{-1})/2$ with $\rho=3+\sqrt{8}$.

\vspace{0.25cm}
\noindent
{\bf Proof of Theorem \ref{thm_Chebyshev_bound}:}
The choice of $\hat \xi^{(0)}_k$ in the initial step corresponds to our convention that 
$$
\beta\Big( 
\begin{matrix}
- \\ -
\end{matrix} \Big | z \Big )\equiv 1=
\frac{1}{2}\hat\xi^{(0)}_0 + \sum\limits_{k\ge 1} \hat\xi^{(0)}_k T_k(4z-1),
$$
thus at this level there is no approximation error: $\xi^{(0)}_k=\hat \xi^{(0)}_k$ for all $k$. 

Let us now justify Steps 1-4. We denote 
$$
f(z):=\beta\Big( 
\begin{matrix}
a_1, \dots, a_{m-1} \\ b_1, \dots, b_{m-1}
\end{matrix} \Big  \vert z
\Big ), \;\;\;
F(z):=\beta\Big( 
\begin{matrix}
a_1, \dots, a_{m} \\ b_1, \dots, b_{m}
\end{matrix} \Big  \vert z
\Big ).
$$
According to \eqref{B_integral_identity} and \eqref{def_ beta}, 
these two functions are related by the equation
\begin{equation*}
z^{A_m} F(z)=\int_0^z x^{A_m-1}(1-x)^{b_m-1} f(x) dx, \;\;\; z\in (0,1),
\end{equation*}
where $A_m:=a_1+a_2+\dots+a_m$. 
We take derivative of both sides of the above equation and obtain a differential equation
\begin{equation}\label{diff_eqn_F_f}
A_m F(z)+z F'(z)=(1-z)^{b_m-1} f(z), \;\;\; z\in (0,1).
\end{equation}
Assume we know the coefficients $\theta_k$ in the Chebyshev expansion 
\begin{equation}\label{eqn_Cheb2}
(1-z)^{b_m-1} f(z)=\theta_0/2+\sum\limits_{k\ge 1} \theta_k T_k(4z-1), \;\;\; 0\le z \le 1/2. 
\end{equation}
Let us denote by $\nu_k$ the coefficients in the Chebyshev expansion 
\begin{equation}\label{eqn_Cheb3}
F'(z)=\nu_0/2+\sum\limits_{k\ge 1} \nu_k T_k(4z-1), \;\;\; 0\le z \le 1/2. 
\end{equation}
Chebyshev polynomials satisfy 
$$
\int T_n(x) dx = \frac{T_{n+1}(x)}{2(n+1)}-
\frac{T_{n-1}(x)}{2(n-1)}, 
$$
which can be established from formulas 8.941 and 8.949.1 in 
\cite{Jeffrey2007}. 
From here we obtain 
$$
\int T_k(4z-1) d z=\frac{1}{8} \Big(\frac{T_{k+1}(4z-1)}{k+1}-\frac{T_{k-1}(4z-1)}{k-1}\Big), \;\;\; k\ge 2. 
$$
We also check directly (using the facts $T_1(x)=x$ and $T_2(x)=2x^2-1$) that 
$$
\int 1 d z=T_1(4z-1)/4, \;\;\; \int T_1(4z-1) d z=T_2(4z-1)/16.
$$
Combining the above formulas, we obtain, after integrating both sides of \eqref{eqn_Cheb3} and rearranging the terms of the series, the following expansion 
\begin{equation}\label{expansion_F2}
F(z)=Q+\sum\limits_{k\ge 1} \frac{\nu_{k-1}-\nu_{k+1}}{8k} T_k(4z-1), \;\;\; 0\le z \le 1/2, 
\end{equation}
for some constant $Q$. 
Next we use formulas 
\begin{align*}
&z  = \frac{1}{4}(T_1(4z-1)+1),\\
&z T_k(4z-1)=\frac{1}{8}(T_{k+1}(4z-1)+2T_k(4z-1)+T_{k-1}(4z-1)), \;\;\; k\ge 1,
\end{align*}
and equations \eqref{diff_eqn_F_f}, \eqref{eqn_Cheb2}, \eqref{eqn_Cheb3} and \eqref{expansion_F2} and conclude that 
\begin{align*}
& A_m \Big[ Q +  \sum\limits_{k\ge 1} \frac{\nu_{k-1}-\nu_{k+1}}{8k} T_k(4z-1) \Big]
 + 
 \frac{\nu_0}{8} \Big(T_1(4z-1)+1\Big)\\&
 +\sum\limits_{k\ge 1}\frac{\nu_k}{8} \Big( T_{k+1}(4z-1) +2T_k(4z-1)+ T_{k-1}(4z-1) \Big) = \theta_0/2+\sum\limits_{k\ge 1} \theta_k T_k(4z-1).
\end{align*}
Comparing the coefficients in front of $T_k$ in the above equation gives us the following identities: 
\begin{align}
\label{eqn_xi1}
&2A_mQ+\frac{1}{4} (\nu_0+\nu_1)=\theta_0,\\
\label{eqn_xi2}
&A_m \frac{\nu_{k-1}-\nu_{k+1}}{8k}+\frac{1}{8} (\nu_{k-1}+2\nu_k+\nu_{k+1})=\theta_k, \;\;\; k\ge 1.   
\end{align}
The equation \eqref{eqn_xi2} can be written in an equivalent form 
\begin{equation}\label{eqn_xi3}
\nu_{k-1}=\frac{1}{1+\frac{A_m}{k}} \Big(8\theta_k-2\nu_k - \nu_{k+1}\Big(1-\frac{A_m}{k} \Big) \Big)
\end{equation}
Now we have a preliminary form of the algorithm for computing $\xi_k^{(m)}$: First we compute recursively $\nu_{k}$ via \eqref{eqn_xi3} 
and then evaluate 
\begin{align}
\xi_0^{(m)}&=2Q=\frac{1}{A_m} \Big( \theta_0-\frac{1}{4}(\nu_0+\nu_1) \Big), \\
\xi_k^{(m)}&=\frac{\nu_{k-1}-\nu_{k+1}}{8k}, \;\;\; 1\le k \le N.  
\end{align}
There are two problems that we need to overcome to make this work. First of all, when performing backward iteration in \eqref{eqn_xi3} we need to start with some values $\nu_N$ and $\nu_{N+1}$ in order to compute $\nu_k$ for $k=N-1,N-2,\dots,1,0$, however we do not know the values of $\nu_N$ and $\nu_{N+1}$. Second, we do not know the values of coefficients $\theta_k$ in the Chebyshev expansion \eqref{eqn_Cheb2}, thus we will need to approximate them and control the resulting error. Let us first consider the former problem. 

We define $\{\hat \nu_k\}_{0\le k <N}$ to be the values obtained by recursion \eqref{eqn_xi3} starting from $\hat \nu_{N}=\hat \nu_{N+1}=0$. Let us estimate the difference $y_k=\nu_k-\hat \nu_k$ for $0\le k <N$.  Note that $y_k$ satisfy the homogeneous recurrence equation
\begin{equation}\label{eqn_y}
y_{k-1}=\frac{1}{1+\frac{A_m}{k}} \Big(-2y_k - y_{k+1}\Big(1-\frac{A_m}{k} \Big) \Big), \;\;\; 0\le k <N 
\end{equation}
started with values $y_N=\nu_N$ and $y_{N+1}=\nu_{N+1}$. We denote 
${\bf u}_k=[y_{k-1},  y_k]^T \in {\r^2}$ and 
$$
Q_k:=\begin{bmatrix} 
-\frac{2}{1+A_m/k} & -\frac{1-A_m/k}{1+A_m/k}  \\
1 & 0
\end{bmatrix}.
$$
Then recursion \eqref{eqn_y} can be rewritten in vector-matrix form as follows
\begin{equation*}
{\bf u}_k = Q_k {\bf u}_{k+1}
\end{equation*}
and we find 
\begin{equation}\label{formula_y}
{\bf u}_k = Q_k Q_{k+1} \dots Q_{N} {\bf u}_{N+1}.
\end{equation}
Let us also define $R_k=V^{-1} Q_k V$, where 
$$ 
V:=\begin{bmatrix} 
-2 & 1  \\
2 & 0
\end{bmatrix}.
$$
One can check that 
$$
R_k=\begin{bmatrix} 
-1 & \frac{1}{2}  \\
0 & -\frac{1-A_m/k}{1+A_m/k}
\end{bmatrix},
$$
so that $R_k$ belongs to class ${\mathcal D}$ of matrices of the form
$$
\begin{bmatrix} 
-1 & \frac{1}{2}  \\
0 & d_{2,2}
\end{bmatrix}
$$
with $|d_{2,2}|<1$. It is easy to show by induction that a product of $l$ matrices from class ${\mathcal D}$ has the form
$$
\begin{bmatrix} 
(-1)^l & d_{1,2}  \\
0 & d_{2,2}
\end{bmatrix}
$$
where $|d_{1,2}|<l/2$ and $|d_{2,2}|<1$, thus the $\lVert \cdot \rVert_{\infty}$ norm of this matrix is less than $1+l$. 
From \eqref{formula_y} we conclude 
$$
{\bf u}_k = V R_k R_{k+1} \dots R_{N} V^{-1} {\bf u}_{N+1} 
$$
so that $\lVert {\bf u}_k \rVert_{\infty}<6(2+N-k)\lVert {\bf u}_{N+1} \rVert_{\infty}$ for $0\le k \le N$ (since $\lVert V \rVert_{\infty}=3$ and $\lVert V^{-1} \rVert_{\infty}=2$). Since $\nu_k$ are coefficients in the Chebyshev expansion 
\eqref{eqn_Cheb3} of the function $F'(z)$ 
that is analytic in $\c \setminus[-1,\infty)$, we know that 
$\nu_k=O(\rho^{-k})$ for any $\rho \in (0,3+\sqrt{8})$. Thus 
$\lVert {\bf u}_{N+1} \rVert_{\infty}=O(\rho^{-N})$ and 
for every $0\le k \le N$ we have $\lVert {\bf u}_k \rVert_{\infty}=O(N \rho^{-N})$ for any $\rho \in (0,3+\sqrt{8})$.
Since $3^{8/5}< (3+\sqrt{8})$, for $\rho \in (3^{8/5}, 3+\sqrt{8})$ we have 
$N \rho^{-N}=O(3^{-8N/5})$ and we conclude that  $y_k=\nu_k-\hat \nu_k=O(3^{-8N/5})$.

Thus we have shown that the backward recursion \eqref{eqn_xi3} is stable and if we start it from $\nu_N=\nu_{N+1}=0$ we will contribute an error of size $O(3^{-8N/5})$. 

Now we need to address the second issue, namely, that we do not know the coefficients $\theta_k$ in the Chebyshev expansion \eqref{eqn_Cheb2}. However, we can approximate $f$ by its truncated Chebyshev series, which will give us an approximation
$$
(1-z)^{b_m-1}f(z) \approx  g(z):=(1-z)^{b_m-1} \Big[
\frac{1}{2}\xi^{(m-1)}_0 + \sum\limits_{k=1}^{N} \xi^{(m-1)}_k T_k(4z-1) \Big], 
$$
and the difference $(1-z)^{b_m-1}f(z) - g(z)$ will be bounded by $O(\rho^{-N})$ for any $\rho \in (0,3+\sqrt{8})$, uniformly on $0\le z \le 1/2$ (this follows from the fact that the coefficients $\xi^{(m-1)}_k$ decay at rate $O(\rho^{-k})$ as $k\to +\infty$). 
The numbers $\eta_k$  computed in \eqref{eqn_etak} are precisely the Chebyshev coefficients of $g(z)$ on $0\le z \le 1/2$. Thus $|\eta_k-\theta_k|=O(\rho^{-N})$ for $0\le k \le N$ and 
any $\rho \in (0,3+\sqrt{8})$. Since we have already demonstrated that the backward iteration \eqref{eqn_xi3} is stable, we see that when we replace $\theta_k$ by $\eta_k$ we contribute an overall error of the size $O(3^{-8N/5})$. 
\qed

\section{Numerical examples}\label{section5}

In this section we present several examples that will illustrate the performance of the two algorithms for computing the generalized beta function 
$$
 \B\Big( 
\begin{matrix}
a_1, \dots , a_n \\ b_1, \dots , b_n
\end{matrix} \Big ). 
$$
We chose to compute this particular function because it is the normalizing constant for the ordered beta distribution 
in \eqref{def_beta_distr_simplex} and it is also used to compute various moments of the generalized beta distribution as in \eqref{Thm1_eqn6}. These two computations were instrumental in \cite{LLMN_2015} and it was the question of efficiently computing the values of the generalized beta function that motivated this research. However, as we explained in section \ref{section4}, the computational complexity of computing the generalized incomplete beta function is the same and the algorithms are very similar.

 For our first example we set $n=3$ and 
$$
[a_1,a_2,a_3]=[0.8, \;0.3,\; 1.5], \;\;\; [b_1, b_2, b_3]=[0.4,\; 1.7,\; 0.8]. 
$$
We compute the value of 
\begin{equation}\label{eq:B_exact_1}
 \B\Big( 
\begin{matrix}
a_1, a_2, a_3 \\ b_1, b_2 , b_3
\end{matrix} \Big )=0.4868940470437834231542713481277\dots
\end{equation}
using a Taylor series expansion approach described in section \ref{subsection_computing_Taylor}. The code was written in Fortran
and we used a multiprecision module MPFUN (written by David Bailey \cite{Bailey_2023}). The Taylor series were truncated at  $N=500$ and we did all computations with precision of 500 digits and then repeated it with the precision of 1000 digits -- this helped us to ensure that all digits displayed in \eqref{eq:B_exact_1} are correct.

Next we take the value in \eqref{eq:B_exact_1} as the ``exact" value and compute approximations using Taylor's method and Chebyshev method with with different values of $N$. Here we have implemented the code in Matlab in standard double precision. We define ${\mathcal E}^T(N)$ to be (the absolute value of) the difference between the exact value in \eqref{eq:B_exact_1} and the value computed using the  algorithm presented in section \ref{subsection_computing_Taylor} with $N$ terms in the Taylor series expansion. Similarly we define ${\mathcal E}^C(N)$ to be (the absolute value of) the difference between the exact value in \eqref{eq:B_exact_1} and the value computed using the  algorithm presented in section \ref{subsection_computing_Chebyshev} with $N$ terms in the Chebyshev series expansion. On Figure \ref{fig1_a} we present the plots of the two errors versus $N$. We see that both errors converge to zero exponentially fast and that the Chebyshev approximation converges faster. This confirms our theoretical results in section
\ref{section4}.

\begin{figure}[h!]
\centering
\subfloat[][Parameter set 1]{\label{fig1_a}\includegraphics[height =6.5cm]{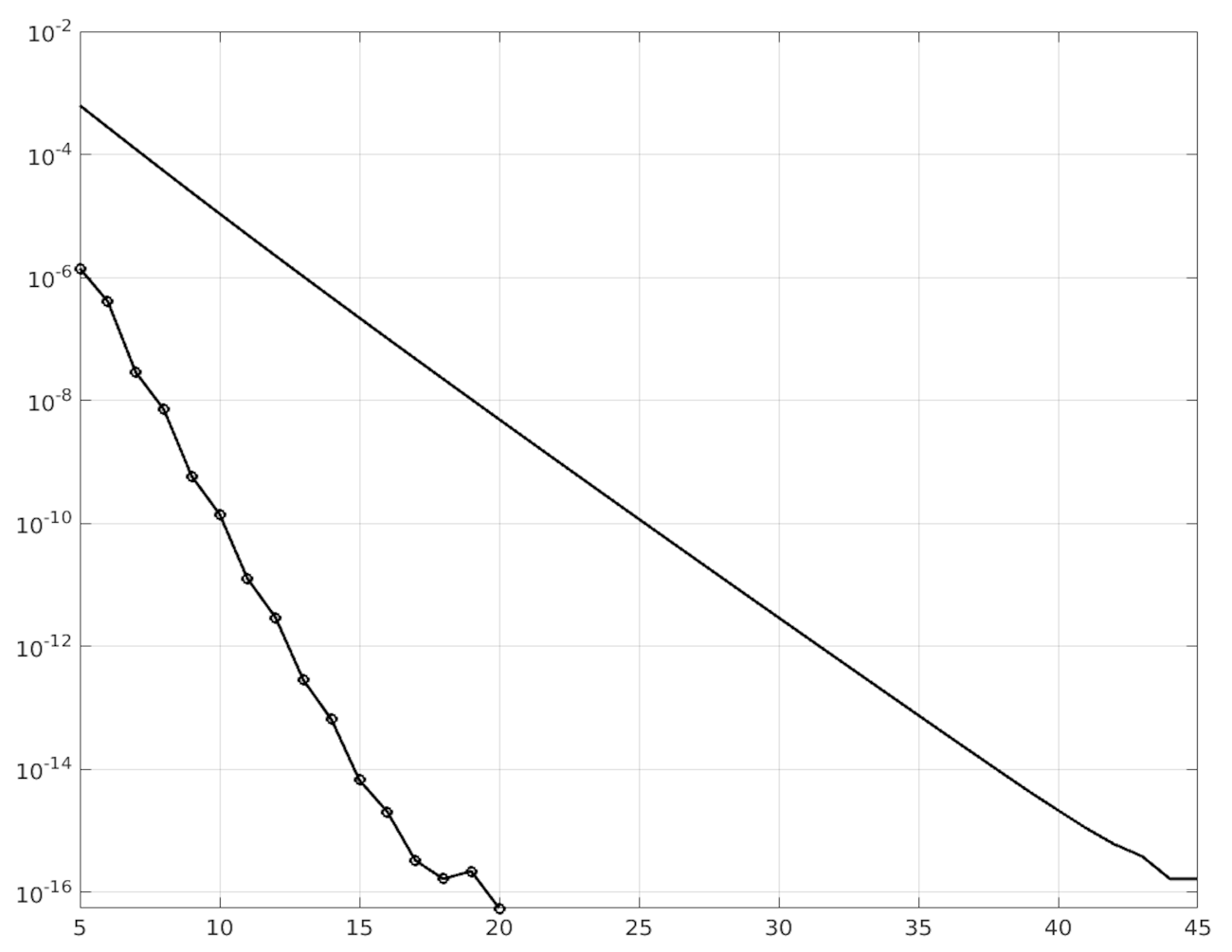}} 
\subfloat[][Parameter set 2]{\label{fig1_b}\includegraphics[height =6.5cm]{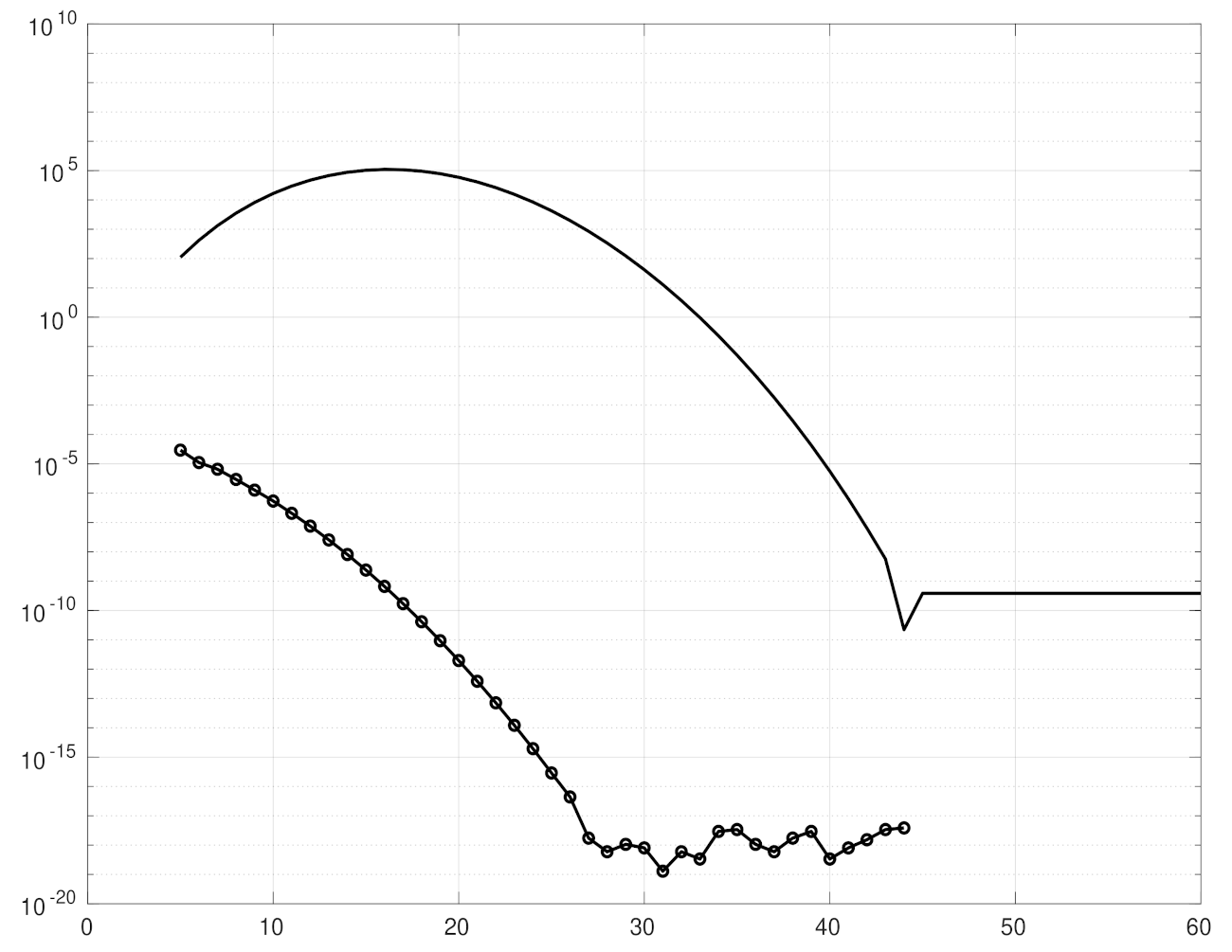}}
\caption{The errors of the approximation ${\mathcal E}^{{ \textnormal{T}}}(N)$ (solid line) and ${\mathcal E}^{{ \textnormal{C}}}(N)$ (line with circles) versus $N$ on the $x$-axis.} 
\label{fig1_error}
\end{figure}

For our second example we consider parameters 
$$
[a_1,a_2,a_3]=[50.8, \;0.3,\; 1.5], \;\;\; [b_1, b_2, b_3]=[0.4,\; 1.7,\; 0.8]. 
$$
Note that we have changed the value of $a_1$ to a relatively large value of $50.8$, while keeping all other parameters the same as in the first example. Again we compute the ``exact" value 
$$
 \B\Big( 
\begin{matrix}
a_1, a_2, a_3 \\ b_1, b_2 , b_3
\end{matrix} \Big )=
10^{-6} \times 9.9752436394601281551585749018468\dots  
$$
using Fortran code with a multiprecision module. On figure \ref{fig1_b} we show the convergence rate of Taylor series and Chebyshev approximations.  We see that the Chebyshev method still converges exponentially fast and is stable, whereas the method based on Taylor expansion struggles in this example: the error first increases to a large value of around $10^6$ and only then decreases. Even with large value of $N=200$ Taylor series method gives us a value of $10^{-6}\times 9.975628\dots$, thus capturing only four correct digits. This confirms our observation at the end of section \ref{subsection_computing_Taylor}, that the Taylor series algorithm is not appropriate in the situation when some of the parameters $a_i$, $b_i$ are large, as there is loss of precision arising due to subtraction of large numbers in 
\eqref{c_k_n_recursive_formula}. In this case one should use either the Chebyshev method of Taylor series method with high precision.

Finally we consider an example with large number of parameters. 
Here we consider an example with 100 parameters and we set 
$$
a_i=\frac{2i-1}{200}, \;\;\; b_i=1-a(i),
$$
for $1\le i \le 100$. 
Again we compute the ``exact" value in Fortran using precision of 500 digits and 500 terms of Taylor series:  
$$ 
\B\Big( 
\begin{matrix}
a_1, a_2, \dots, a_{100} \\ b_1, b_2, \dots, b_{100}
\end{matrix} \Big )=10^{-33} \times  4.2217553528914884124401921234246\dots
$$
Then we try to compute the same value using our Matlab code. The method based on Taylor expansion with $N=50$ terms gives us  
$$
10^{-33} \times 4.221755352891131\dots
$$
and the method based on Chebyshev expansion with $N=20$ terms gives us a comparably accurate result
$$
10^{-33} \times  4.221755352891663\dots. 
$$
We see that both methods work well in this case. 

We provide both Matlab programs (Taylor and Chebyshev methods) for computing the generalized beta function  
$
\B\Big( 
\begin{matrix}
a_1, a_2, \dots, a_{n} \\ b_1, b_2, \dots, b_{n}
\end{matrix} \Big )
$
on \url{https://kuznetsov.mathstats.yorku.ca/code/}.

\section*{Acknowledgements}
Research of A.K. was supported by the Natural Sciences and Engineering Research Council of Canada. 

%
%\bibliographystyle{abbrv}
%\bibliography{references}

%****************************************************************************************************************
%****************************************************************************************************************
%****************************************************************************************************************

\end{document}